\documentclass[runningheads]{llncs}
\usepackage{amssymb}
\usepackage{textcomp}
\usepackage{amsmath}
\usepackage{mathtools}

\usepackage{amsthm}
\usepackage{algorithm}
\usepackage[noend]{algpseudocode}
\usepackage{comment}
\usepackage{graphicx}
\usepackage{color}
\usepackage{tikz}
\usepackage{cleveref}
\DeclareMathOperator{\argmin}{arg\,min}
\DeclareMathOperator{\dist}{dist}

\begin{document}

\title{Time Efficient Implementation for Online $k$-server Problem on Trees
}
%
%
\author{Kamil Khadiev\inst{1} \and
Maxim Yagafarov\inst{1}}
\authorrunning{K. Khadiev, et al.}
%
\institute{Kazan Federal University, Kazan, Russia}
\maketitle              
\begin{abstract}
We consider online algorithms for the $k$-server problem on trees of size $n$. 
Chrobak and Larmore proposed a $k$-competitive algorithm for this problem that has the optimal competitive ratio. However, the existing implementations have  $O\left(k^2 + k\cdot \log n\right)$ or $O\left(k(\log n)^2\right)$ time complexity for processing a query, where $n$ is the number of nodes. We propose a new time-efficient implementation of this algorithm that has $O(n)$ time complexity for preprocessing and $O\left(k\log k\right)$ time for processing a query. The new algorithm is faster than both existing algorithms and the time complexity for query processing does not depend on the tree size. 
\end{abstract}
\section{Introduction}
\label{sec:intro}
Online optimization is a field of optimization theory that deals with optimization problems not know the future \cite{k2016}. An online algorithm reads an input piece by piece and returns an answer piece by piece immediately, even if the optimal answer can depend on future pieces of the input. The goal is to return an answer that minimizes an objective function (the cost of the output). The most standard method to define the effectiveness of an online algorithm is the competitive ratio
\cite{st85,kmrs86}.
 That is the worst-case ratio between the cost of the solution found by the algorithm and the cost of an optimal solution. If the ratio is $c$, then the online algorithm is called $c$-competitive. %
In the general setting, online algorithms have unlimited computational power. Nevertheless, many papers consider them with different restrictions. Some of them are restrictions on memory 
\cite{bk2009,blm2015,kkm2018,aakv2018,gs93,h95,kk2019disj,k2021,kl2020,kk2019online2w,kk2019,kkkmry2023,kk2022,kkzmkry2022},
others are restrictions on time complexity \cite{fnn2006,rbm2013,kkmsy2022,ky2022}.
This paper focuses on efficient online algorithms in terms of time complexity.
One of the well-known online minimization problems is the $k$-server problem on trees \cite{cl91}. Other related well-known problems are the matching problem, r-gathering problem, and facility assignment problem \cite{kp93,an2015,ark2020}.
The $k$-server problem on trees is the following. We have a tree with $n$ nodes and $k$ servers which are in some nodes. In online fashion, we receive queries that are nodes of the tree. Servers move by the graph and one of the servers should come to the query node. Our goal is a minimization of the total moving distance for all servers in all queries.
There is a $k$-competitive deterministic algorithm for the  $k$-server problem on trees \cite{cl91}, and the algorithm has the best competitive ratio. The expected competitive ratio for a best-known randomized algorithm \cite{bbmn2011,bbmn2015} is $O(\log^3 n\log^2 k)$. In this paper, we are focused on the deterministic one. So, the competitive ratio of the deterministic algorithm is the best possible. At the same time, the naive implementation has $O(n)$ time complexity for each query and preprocessing. There is a time-efficient algorithm for general graphs \cite{rbm2013} that uses a min-cost-max-flow algorithm, but it is too slow in the case of a tree. 
In the case of a tree, there are two efficient algorithms. The first one was introduced in \cite{kkmsy2022}. It has time complexity $O(n\log n)$ for preprocessing and $O\left(k^2+k\log n\right)$ for query processing. Another one was presented in \cite{ky2022}. It has time complexity $O(n)$ for preprocessing and $O\left(k(\log n)^2\right)$ for query processing.

We suggest a new time-efficient implementation of the algorithm from \cite{cl91}. It has $O(n)$ time complexity for preprocessing and $O\left(k\log k\right)$ for processing a query. It is based on the compression of tree technique called virtual tree and fast algorithms for computing Lowest Common Ancestor (LCA) \cite{bv93,bfc2000} and Level Ancestor (LA) Problem \cite{ah2000,bf2004,d1991}.
The presented algorithm has a better time complexity of query processing compared to all existing algorithms. The query processing is asymptotically better than both algorithms. At the same time, the query processing complexity of our algorithm does not depend on $n$. If $k$ is constant, then the complexity is $O(1)$.
The time complexity of preprocessing is the same as for the Algorithm from \cite{ky2022}, and better than the preprocessing complexity of the Algorithm from \cite{kkmsy2022}.


\vspace{-0.3cm}
\section{Preliminaries}
\label{sec:prelims}
\subsection{Online algorithms}
{\bf An online minimization problem} consists of a set $\cal{I}$ of inputs and a cost function. Each input $I = (x_1, \dots , x_M)$ is a sequence of requests, where $M=|I|$ is the length of the input. Furthermore, a set of feasible outputs (or solutions) $\widetilde{O}(I)$ is associated with each $I$; an output is a sequence of answers $O = (y_1, \dots, y_M)$. The cost function assigns a positive real value $cost(I, O)$ to $I\in{ \cal I}$ and $O\in\widetilde{O}(I)$. An optimal  solution for $I\in{\cal I}$ is $O_{opt}(I)=\argmin_{O\in\widetilde{O}(I)}cost(I,O)$.

Let us define an online algorithm for this problem.
{\bf A deterministic online algorithm}  $A$ computes the output sequence $A(I) = (y_1,\dots , y_M)$ such that $y_i$ is computed based on $x_1, \dots , x_i$.  
We say that $A$ is $c$-{\em competitive} if there exists a constant $\alpha\geq 0$ such that, for every $n$ and for any input $I$ of size $n$, we have: $cost(I,A(I)) \leq c \cdot cost(I,O_{Opt}(I)) + \alpha$. The minimal $c$ that satisfies the previous condition is called the {\bf competitive ratio} of $A$. 
\vspace{-0.5cm}
\subsection{Rooted Trees}\label{sec:lca}
Consider a rooted tree $T=(V,E)$, where $V$ is the set of nodes/vertices, and $E$ is the set of undirected edges. Let $n=|V|$ be the number of nodes, or equivalently the size of the tree. We denote by $1$ the root of the tree.
A path $P$ is a sequence of nodes $(v_1,\dots,v_h)$ that are connected by edges, i.e. $(v_i,v_{i+1})\in E$ for all $i\in\{1,\dots,h-1\}$, such that there are no duplicates among $v_1,\dots,v_h$. Here $h-1$ is the length of the path. Between any two nodes $v$ and $u$ on the tree, there is a unique path. The distance $\dist(v,u)$ is the length of this path.
For each node $v$ we can define a parent node $\textsc{Parent}(v)$ which is the first node on the unique path from $v$ to root $1$.  We have $\dist\left(1,\textsc{Parent}(v)\right)+1=\dist(1,v)$. Additionally, we can define the set of children $\textsc{Children}(v)=\{u: \textsc{Parent}(u)=v\}$. Any node $y$ on the unique path from root $1$ to node $v$ is an ancestor of node $v$. 


{\bf Lowest Common Ancestor (LCA).}
Given two nodes $u$ and $v$ of a rooted tree, the Lowest Common Ancestor is the node $w$ such that $w$ is an ancestor of both $u$ and $v$, and $w$ is the closest one to $u$ and $v$ among all such ancestors. The following result is well-known.

\begin{lemma}[\cite{bfc2000}]\label{lm:lca}
There is an algorithm for the LCA problem with the following properties:
    (i) The time complexity of the preprocessing step is $O(n)$;
    (ii)The time complexity of computing LCA for two nodes is $O(1)$.
\end{lemma}

We call $\textsc{LCA\_Preprocessing}()$ the subroutine that does the preprocessing for the algorithm and $\textsc{LCA}(u,v)$ that computes the LCA of two nodes $u$ and $v$.


{\bf Level Ancestor (LA) Problem.}
Given a node $v$ of a rooted tree and a distance $d$, we should find a vertex $u$ that is the ancestor of the node $v$ and has distance $d$ from the root. The following result is well-known. There are several algorithms with required complexity \cite{ah2000,bf2004,d1991}.

\begin{lemma}[\cite{ah2000,bf2004,d1991}]\label{lm:la}
There is an algorithm for the LA problem with the following properties:
    (i) The time complexity of the preprocessing step is $O(n)$;
    (ii)The time complexity of computing the required node is $O(1)$.
\end{lemma}

We call $\textsc{LA\_Preprocessing}()$ the subroutine that does the preprocessing for the algorithm and $\textsc{LA}(v,d)$ that computes the ancestor of the node $v$ and has distance $d$ from the root.
\vspace{-0.5cm}
\subsection{$k$-server Problem on Trees}
Let $T=(V,E)$ be a rooted tree, and we are given $k$ servers that can move among nodes of $T$. The tree and initial positions of the servers are given. At each time slot, a request $q\in V$
appears as an input. We have to ``serve'' this request, that is, to choose one of the $k$ servers and
move it to $q$. The other servers are also allowed to move. The cost function is the distance by which we move the servers.
In other words, if before the request, the servers are at nodes $v_1,\dots,v_k$ and after the request they are at $v'_1,\dots,v'_k$, then $q\in\{v'_1,\dots,v'_k\}$ and the cost of the move is $\sum_{i=1}^k \dist(v_i,v'_i)$. The cost of a sequence of requests is the sum of the costs of serving all requests. The problem is to design a strategy that minimizes the cost
of servicing a sequence of requests given online.

\section{Algorithm}
Firstly, we describe Chrobak-Larmore's $k$-competitive algorithm for $k$-server problem on trees from \cite{cl91}.
Assume that we have a request on a node $q$, and the servers are on the nodes $v_1,\dots,v_k$.
We say that a server $i$ is {\em active} if there are no other servers on the path from $v_i$ to $q$. If there are several servers in a node, then we assume that the server with the minimal id in initial enumeration is {\em active} and others are not {\em active}.
 In each phase, we move every {\em active} server one step towards the node $q$. After each phase, the set of {\em active} servers can change.
 We repeat this phase (moving the active servers) until one of the servers reaches the queried node $q$.

In this section, we present an effective implementation of Chrobak-Larmore's algorithm with preprocessing. The preprocessing part is done once and has $O(n)$ time complexity (Theorem \ref{th:preproc}). The request processing part is done for each request and has $O\left(k\log k\right)$ time complexity (Theorem \ref{th:query-proc}).
 
\subsection{Preprocessing}
We do the following steps for the preprocessing 
:
    %
    (i) We invoke preprocessing for the LCA algorithm (Section \ref{sec:lca}).
    %
     (ii) We invoke preprocessing for the LA problem (Section \ref{sec:lca}). 
    %
     (iii) For each node $v$ we compute the distance from the root to $v$, i.e. $\dist(1,v)$. This can be done using a depth-first search (DFS) algorithm \cite{cormen2001}.
During the DFS algorithm, we store in and out time for a node. We use a common counter variable as a timer and increase it in each event of entering or leaving a node. For a node $v$ we store values of the timer on entering the node event in $t_{in}(v)$ and on leaving the node event in $t_{out}(v)$. We use two arrays for storing the values of $t_{in}$ and $t_{out}$ for every node. The implementation of step is $\textsc{ComputeDistance}()$ procedure which implementation is presented in Appendix \ref{apx:dist}.
The implementation of the whole preprocessing procedure is presented in Algorithm \ref{alg:preproc}, and the time complexity is discussed in Theorem \ref{th:preproc}.
%
%
%
%
%
\begin{algorithm}[ht]
    \caption{$\textsc{Preprocessing}$. The preprocessing procedure.} \label{alg:preproc}
    \begin{algorithmic}
        \State $\textsc{LCA\_Preprocessing()}$
        \State $\textsc{LA\_Preprocessing()}$
        \State $\dist(1,1)\gets 0$
        \State $\textsc{ComputeDistance(1)}$
    \end{algorithmic}
\end{algorithm}
\begin{theorem}\label{th:preproc}
The preprocessing has time complexity $O(n)$.
\end{theorem}
\begin{proof}
The time complexity of the preprocessing phase is $O(n)$ for LCA, $O(n)$ for the LA problem, and $O(n)$ for $\textsc{ComputeDistance}()$ as a complexity of depth-first search algorithm \cite{cormen2001}. Therefore, the total time complexity is $O(n)$.
\end{proof}

\subsection{Query Processing}
Assume that we have a query on a node $q$, and the servers are on the nodes $v_1,\dots,v_k$.
As a first step, we construct a Virtual Tree. We describe this process in the Virtual Tree subsection of the current Section. It is a compressed version of the original tree that contains the query node, server nodes, and their LCA nodes. We show that we can use only this tree for processing the query. The size of the compressed tree is $O(k)$. So, it allows us to obtain a small time complexity for processing the query. The whole algorithm of processing a query is presented in Section \ref{sec:query-algorithm} (Algorithm for Query Processing).
\subsubsection{Virtual Tree}\label{sec:virtual-tree}
Let $V^s=\{v_1,\dots,v_k,q\}$ be a set of nodes that contains servers and the query. Let $V^{lca}=\{LCA(v,w): v,w\in V^s\}$ be a set of all LCA nodes for pairs from $V^s$.
 Let us sort nodes from $V^s$ according to the DFS traversal order. In other words, we have $v_{i_1},\dots,v_{i_{k+1}}$ such that $t_{in}(v_{i_j})\leq t_{in}(v_{i_{j+1}})$. Let $V^{olca}=\{LCA(v_{i_j},v_{i_{j+1}}): 1\leq j\leq k\}$ be set of all LCA nodes for sequential pairs of server nodes in the DFS traversal order. Let us show two properties:
 (i) $V^{lca}= V^s\cup V^{olca}$ (in Lemma \ref{lm:lca-join});
 (ii) $|V^{lca}|\leq  2|V^s|$  (in Lemma \ref{lm:lca-size}).
These two properties show that we can easily compute the set $V^{lca}$, and it is not much larger than the original set $V^s$.

\begin{lemma}\label{lm:lca-join}
$V^{lca}= V^s\cup V^{olca}$.
\end{lemma}
\begin{proof}
By definition, $ V^{olca}\subseteq V^{lca}$.  At the same time, $LCA(v,v)=v$. Therefore, each node $v\in V^s$ belongs to $V^{lca}$. So, $V^s\subseteq V^{lca}$. Hence, $ V^s\cup V^{olca}\subseteq V^{lca}$. To finish the proof, we should show that $V^{lca}\subseteq V^s\cup V^{olca}$.
Assume that we have a node $par\in V^{lca}$ such that $par\not\in  V^s\cup V^{olca}$. In other words, we have $u,v\in V^s$ such that $LCA(u,v)\not\in  V^s\cup V^{olca}$. W.l.o.g. we can assume that $u$ has a smaller index than the index of $v$ in the order $i_1,\dots,i_{k+1}$.  Let $T_u$ be a sub-tree of $par$ and $u\in T_u$. Let $T_v$ be a sub-tree of $par$ and $v\in T_v$. Let $T_{m}$ be a set of all other sub-trees of $par$ that are between $T_u$ and $T_v$ in the DFS traversal order. See Figure \ref{fig:lca-proof} for illustration.

\begin{figure}
    \centering
    \includegraphics[width=0.35\textwidth]{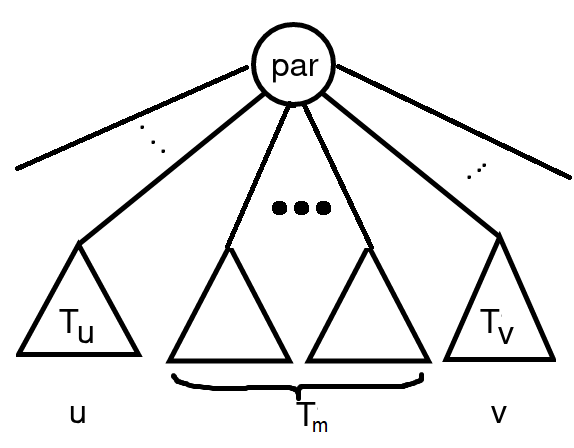}
    \caption{Subtrees of the $par=LCA(u,v)$ node.\label{fig:lca-proof}}
\end{figure}

 Let $j$ be the minimal index such that $v_{i_j}\in T_v$. Therefore, $v_{i_{j-1}}\in T_u \cup T_m$ because we have at least $u$ from these sub-trees that has smaller index in the DFS traversal order. The nodes $v_{i_j}$ and $v_{i_{j-1}}$ are in different sub-trees of $par$. Hence, $par=LCA(v_{i_{j-1}}, v_{i_{j}})$, and $par\in V^{olca}$. We obtain a contradiction with the assumption. Therefore, any $par\in V^{lca}$ belongs to $V^s\cup V^{olca}$. So, we have $V^{lca}\subseteq V^s\cup V^{olca}$. Finally, we obtain the claim of the lemma.
\end{proof}
\begin{lemma}\label{lm:lca-size}
$|V^s|\leq |V^{lca}|\leq  2|V^s|$.
\end{lemma}
\begin{proof}
Due to construction of $V^{olca}$, we have $|V^{olca}|\leq |V^s|-1$.
Due to Lemma \ref{lm:lca-join}, we have $V^{lca}\subseteq V^s\cup V^{olca}$. Therefore, $|V^{lca}|\leq |V^s|+ |V^{olca}|\leq |V^s|+|V^s|-1\leq 2|V^s|$. At the same time, $V^s\subseteq V^{lca}$. So, $|V^s|\leq |V^{lca}|$.
\end{proof}

A virtual tree $T^{virt}$ is a weighted tree constructed by the tree $T$ and the set $V^s$ that has the following structure. The set of the virtual tree's nodes is $V^{lca}$. There is an edge $(u,v)$ in $T^{virt}$ iff there is no $x\in V^{lca}$ such that $x\neq u, x\neq v$ and $dist(x,u)+dist(x,v)=dist(u,v)$. Note that $dist$ is the distance between nodes in the original tree $T$. In other words, there is no $x\in V^{lca}$ such that it is on the path between $u$ and $v$ in the original tree $T$. The weight of the edge $(u,v)$ in $T^{virt}$ is $w^{virt}(u,v)=dist(u,v)$. Let us present an example of the virtual tree in Figure \ref{fig:virt-tree}.

\begin{figure}
    \centering
    \includegraphics[width=\textwidth]{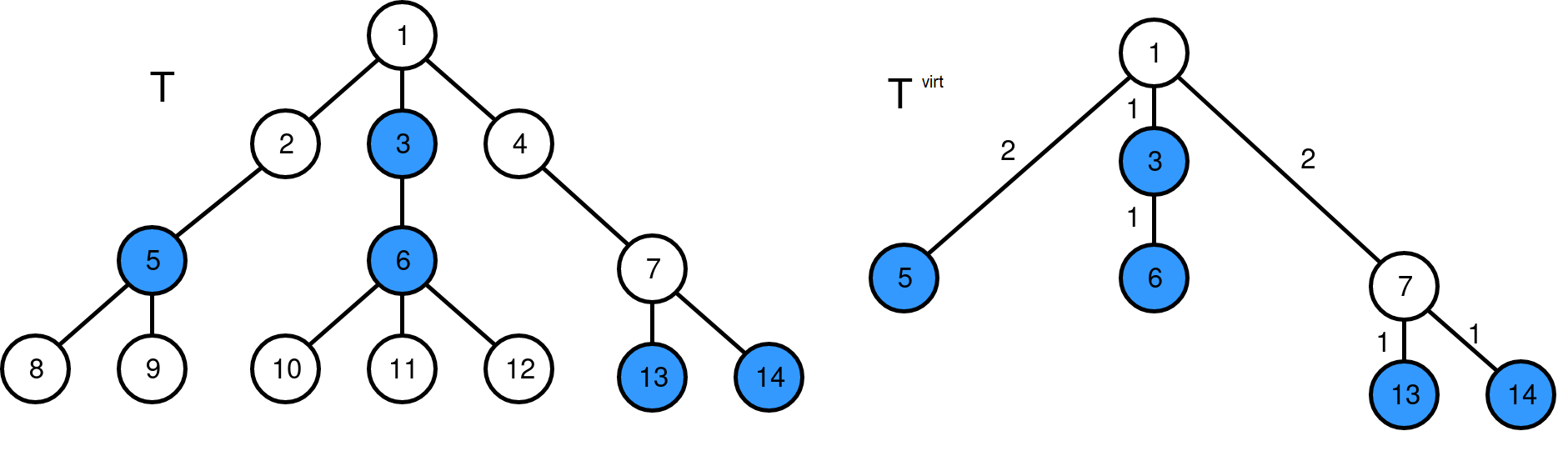}
    \caption{The virtual tree $T^{virt}$ for the tree $T$ and set $V^s=\{3,5,6,13,14\}$.\label{fig:virt-tree}}
\end{figure}

Let us present the algorithm for constructing the virtual tree $T^{virt}$.
\begin{itemize}
\item[] \textbf{Step 1.} Let us sort all nodes of $ V^{lca}$ in the DFS traversal order. In the algorithm, we assume that $ V^{lca}$ is sorted.
\item[] \textbf{Step 2.} Let us use a stack ${\cal S}$ for nodes from  $ V^{lca}$. Initially, we push the first node $v$ from $ V^{lca}$.
\item[] \textbf{Step 3.} For each node $v$ from  $ V^{lca}$ we do Steps 4, 5 and 6.
\item[] \textbf{Step 4.} We remove all nodes from the top of ${\cal S}$ that are not ancestors of $v$. We can check the property in constant time using arrays $t_{in}$ and $t_{out}$. A node $u$ is an ancestor of a node $v$ iff $t_{in}(u)\leq t_{in}(v)$ and $t_{out}(v)\leq t_{out}(u)$. We assume that we have $\textsc{IsAncestor}(u,v)$ procedure that returns $True$ iff $u$ is an ancestor of $v$. 
\item[] \textbf{Step 5.} If after Step 4 the stack is not empty, then we take the top node (without removing it from the stack). Let it be $u$. Then, we add an edge $(u,v)$ to the tree $T^{virt}$. If the stack is empty, then we do nothing.
\item[] \textbf{Step 6.} We push the node $v$ to the stack.
\end{itemize}

Let the algorithm be $\textsc{ConstructVirtualTree}(T,v_1,\dots,v_k,q)$ procedure and its implementation is presented in Algorithm \ref{alg:virt-tree}.
We assume that we have the following procedures:
\begin{itemize}
\item  $\textsc{AddEdge}(T^{virt},u,v)$ procedure for adding the edge $(u,v)$ to the tree $T^{virt}$.
\item $\textsc{InitStack}({\cal S})$ procedure for initialization an empty stack.
\item $\textsc{Push}({\cal S},v)$ procedure for putting a node $v$ to the top of the stack.
\item $\textsc{Pop}({\cal S})$ procedure for removing a node from the top of the stack.
\item $\textsc{Pick}({\cal S})$ procedure that returns the top element of the stack but does not remove it.
\item $\textsc{IsEmpty}({\cal S})$ procedure that returns $True$ iff  the stack is empty; and otherwise $False$.
\end{itemize}

 \begin{algorithm}[ht]
    \caption{$\textsc{ConstructVirtualTree}(T,v_1,\dots,v_k,q)$ procedure for constructing virtual tree $T^{virt}$ for the original tree $T$ and the set of server nodes $V^s$} \label{alg:virt-tree}
    \begin{algorithmic}
        \State $(v_{i_1},\dots,v_{i_{k+1}})\gets\textsc{Sort}(v_1,\dots,v_k,q)$\Comment{Sorting $V^s$ in the DFS traversal order} 
        \State $V^{olca}\gets\{\}$\Comment{Initially $V^{olca}$ is empty} 
        
        \For{$j\in\{1,\dots,k\}$}
            \State  $V^{olca}\gets V^{olca}\cup\{\textsc{LCA}(v_{i_j},v_{i_{j+1}})\}$        
        \EndFor
          \State $V^{lca}\gets V^{s}\cup V^{olca}$ 
          \State $(b_1,\dots b_{|V^{lca}|})\gets\textsc{Sort}(V^{lca})$
          \State $\textsc{InitStack}({\cal S})$
          \State $\textsc{Push}({\cal S},v_{b_1})$
          \For{$j\in\{2,\dots,|V^{lca}|\}$}
          \While{$\textsc{IsEmpty}({\cal S})=False$ and $\textsc{IsAncestor}(\textsc{Pick}({\cal S}),v_{b_j})=False$}
			\State    $\textsc{Pop}({\cal S})$       
          \EndWhile
          \If{$\textsc{IsEmpty}({\cal S})=False$ }
          \State $\textsc{AddEdge}(T^{virt},\textsc{Pick}({\cal S}),v_{b_j})$   
          \EndIf
          \State $\textsc{Push}({\cal S},v_{b_j})$
          \EndFor
    \end{algorithmic}
\end{algorithm}


The complexity of the algorithm is presented in the next lemma.
\begin{lemma}
The time complexity of Algorithm \ref{alg:virt-tree} is $O(k\log k)$.
\end{lemma}
\begin{proof}
Sorting $k$ nodes in $V^s$ has $O(k\log k)$  time complexity. The procedure $\textsc{LCA}$ has $O(1)$ time complexity due to Lemma \ref{lm:lca}. We can implement sets $V^{lca}$ and $V^{olca}$ as Self-Balanced Search Tree \cite{cormen2001}. Due to Lemma \ref{lm:lca-size}, $|V^s|,|V^{lca}|,|V^{olca}|=O(k)$.
 So, the complexity of adding an element to a set has $O(\log k)$ time complexity \cite{cormen2001}. Therefore, constructing $V^{lca}$ has $O(k\log k)$ time complexity.
The complexity of Step 1 is $O(k\log k)$ because of the complexity of sorting.
We add each node of $V^{lca}$ to the stack once and remove it once. Procedures $\textsc{IsAncestor}$, $\textsc{IsEmpty}$, $\textsc{Pop}$, $\textsc{Pick}$, $\textsc{AddEdge}$, and $\textsc{Push}$ have $O(1)$ time complexity. Therefore, the total time complexity of the for-loop is $O(k)$.
So, the total time complexity of the algorithm is $O(k\log k)+O(k)=O(k\log k)$.
\end{proof}
\subsubsection{Algorithm for Query Processing}\label{sec:query-algorithm}

We have several events. Each event is moving an active server or deactivating an active server. Note that during such events, active servers can be only in nodes from $V^{lca}$. Only at the deactivating moment, a server can move to a node that does not belong to $V^{lca}$.  So, we implement an algorithm that moves servers by the virtual tree $T^{virt}$. At the same moment, we have a procedure $\textsc{Move}(u,z)$ that moves a server from a node $v$ to the distance $z$ in the direction to the query $q$ in the original tree $T$. We discuss the $\textsc{Move}(u,z)$ procedure in Section \ref{sec:move}.

Let us discuss the algorithm. Let us fix the node $q$ as a root of the virtual tree $T^{virt}$. 
The algorithm contains two main phases.

\textbf{Phase 1.} For each node $v\in V^{lca}$, we compute two values that are a node $\textsc{Closest}(v)$ and an integer $\textsc{DistToClosest}(v)$ that are
\begin{itemize}
\item $\textsc{Closest}(v)$ is the closest server to the node $v$ from the subtree with the root $v$. If there are several such servers, then $\textsc{Closest}(v)$ is the server with the smallest id in initial enumeration among them.  In other words, it is the server that comes to the node $v$ if we have the only subtree with the root $v$. 
\item $\textsc{DistToClosest}(v)$ is the distance from $\textsc{Closest}(v)$ to $v$
\end{itemize}
We can compute these two values for each node using the depth-first search algorithm.
Assume that we are processing $v$ and already invoked the procedure for children and we know values of $\textsc{Closest}$ and $\textsc{DistToClosest}$ for them. In that case $\textsc{DistToClosest}(v)=min\{\textsc{DistToClosest}(u)+w^{virt}(u,v): u\in\textsc{Children}(v)\}$, and $\textsc{Closest}(v)=\textsc{Closest}(u)$ for the corresponding child. Note that if the node contains a server ($v\in \{v_1,\dots,v_k\}$), then $\textsc{Closest}(v)=v$ and $\textsc{DistToClosest}(v)=0$. Note that we can implement the set of servers using a Self-balanced Binary Search Tree (for example, Red-Black tree \cite{cormen2001}) and check the property with $O(\log k)$ time complexity.   Let the implementation of the phase be a $\textsc{ClosestComputing}(v)$  procedure that is presented in Algorithm \ref{alg:phase1}. We assume that we have $\textsc{Id}(v_i)$ function that returns id in initial enumeration for a server $v_i$.

\begin{algorithm}[ht]
    \caption{$\textsc{ClosestComputing}(v)$ procedure for the first phase. Here $v\in V^{lca}$ is the processed node.} \label{alg:phase1}
    \begin{algorithmic}
        
        \For{$u\in \textsc{Children}(v)$}
          \State $\textsc{ClosestComputing}(u)$          
        \EndFor
        \If{$v\in \{v_1,\dots,v_k\}$}
        \State $\textsc{Closest}(v)\gets v$
        \State $\textsc{DistToClosest}(v)\gets 0$
        \Else
        \State $v_c\gets NULL$
        \State $d_c\gets NULL$
         \For{$u\in \textsc{Children}(v)$}
         \If{$d_c=NULL$ or $d_c>\textsc{DistToClosest}(u)+w^{virt}(u,v)$ or ($d_c=\textsc{DistToClosest}(u)+w^{virt}(u,v)$ and $\textsc{Id}(v_c)>\textsc{Id}(\textsc{Closest}(u))$)}
         
          \State $d_c\gets \textsc{DistToClosest}(u)+w^{virt}(u,v)$ 
          \State $v_c\gets \textsc{Closest}(u)$
          \EndIf         
        \EndFor
       	\State $\textsc{Closest}(v)\gets v_c$
       	\State $\textsc{DistToClosest}(v)\gets d_c$
        \EndIf
    \end{algorithmic}
\end{algorithm}

\textbf{Phase 2.} In this phase, we compute new positions of servers. This phase is also based on the depth-first search algorithm.
 We start with the root node $r$. Note that we fix the node $q$ as a root of the virtual tree. We know that the $\textsc{Closest}(r)$ server comes to this node in $\textsc{DistToClosest}(r)$, and it processes the query. It means that 
\begin{itemize}
 \item the new position of server $\textsc{Closest}(r)$ is $q$;
 \item any other servers  become inactive at most after $\textsc{DistToClosest}(r)$ steps. 
\end{itemize}
Moreover, we can say that any server $\textsc{Closest}(u)$, for $u\in \textsc{Children}(r)$  become inactive exactly after $\textsc{DistToClosest}(r)$ steps. It is correct because no other server comes to $u$ in their subtree, but  $\textsc{Closest}(r)$ comes faster to $q$ than others. So, we can move all servers  $v\in\{\textsc{Closest}(u): u\in \textsc{Children}(r)\}\backslash\{\textsc{Closest}(r)\}$ to the distance $\textsc{DistToClosest}(r)$, and it is the final position of these servers.

Let us look at the children of a node $u$, where  $u\in\textsc{Children}(r)$. There are two options:
\begin{itemize}
\item all servers from children becomes inactive because of $\textsc{Closest}(r)$ server;
\item all servers from children becomes inactive because of $\textsc{Closest}(u)$ server. In that case, the $\textsc{Closest}(u)$ server can come to $u$ faster than in $\textsc{DistToClosest}(r)$ steps. At the same time, the way of the $\textsc{Closest}(u)$ server to $r$ takes more than $\textsc{DistToClosest}(r)$ steps. 
\end{itemize} 

So, we take $b\gets \min\{\textsc{Closest}(r), \textsc{Closest}(u)\}$, and be sure that any child of $u$ becomes inactive in $b$ steps.

Let us describe the general case. Assume that we process a node $v$ and $b$ is the number of steps for deactivating $\textsc{Closest}(v)$.

Firstly, we update $b\gets \min\{b,\textsc{DistToClosest}(v)\}$. Then, we consider all children of $v$ that are $u\in \textsc{Children}(v)$. If $\textsc{Closest}(u)=\textsc{Closest}(v)$, i.e. this server comes to the node, then we ignore it (we already moved them on processing the parent of $v$). For other servers, we move $\textsc{Closest}(u)$ to $b$ steps.

Finally, we invoke the same function for all children of $v$ and the new value of $b$.

 Let the implementation of the phase be a procedure $\textsc{UpdatePositions}(v,b)$ that is presented in Algorithm \ref{alg:phase2}.
 \begin{algorithm}[ht]
    \caption{$\textsc{UpdatePositions}(v,b)$ procedure for second phase. Here $v$ is the processed node, and in $b$ steps $\textsc{Closest}(v)$ becomes inactive.} \label{alg:phase2}
    \begin{algorithmic}
        \State $b\gets \min\{b,\textsc{DistToClosest}(v)\}$
       
        \For{$u\in \textsc{Children}(v)$}
          \If{$\textsc{Closest}(v)\neq \textsc{Closest}(u)}$
           \State $\textsc{Move}(\textsc{Closest}(u),b)$
        \EndIf                 
        \EndFor
      \For{$u\in \textsc{Children}(v)$}          
           \State $\textsc{UpdatePositions}(u,b)$                  
        \EndFor
    \end{algorithmic}
\end{algorithm}

Let us present the final algorithm for processing a query $q$ in Algorithm \ref{alg:query}. Assume that $\textsc{Root}(T^{virt})$ is the root node of the virtual tree $T^{virt}$.

 \begin{algorithm}[ht]
    \caption{$\textsc{ProcessingAQuery}(q,T, v_1,\dots,v_k)$ procedure for processing a query $q$ on the tree $T$ and initial positions of servers $v_1,\dots,v_k$.} \label{alg:query}
    \begin{algorithmic}
        \State $\textsc{ConstructVirtualTree}(T,v_1,\dots,v_k,q)$
        \State $r\gets\textsc{Root}(T^{virt})$	
        \State $\textsc{ClosestComputing}(r)$
        \State $\textsc{Move}(\textsc{Closest}(r),\textsc{DistToClosest}(r))$
        \State $\textsc{UpdatePositions}(r,\textsc{DistToClosest}(r))$
    \end{algorithmic}
\end{algorithm}
\vspace{-0.2cm}
\subsection{Moving a Server}\label{sec:move}
\vspace{-0.1cm}
We now consider the following problem: given a server $v$ and a distance $z$, how to efficiently compute the new position of the server after moving it $z$ steps towards $q$. As a subroutine, we use the solution of the Level Ancestor (LA) Problem from Section \ref{sec:lca}.


Let $l=LCA(v,q)$. If $\dist(l,v)\geq z$, then the result node is on the path between $v$ and $l$. Note that $l$ is an ancestor of $v$, therefore, the target node is also an ancestor of $v$. We can say that the result node is on the distance $d=dist(1,v)-z$ from the root node. So, the target node is the result of $\textsc{LA}(v,dist(1,v)-z)$ procedure from Section \ref{sec:lca}. Otherwise, we should move the server first to $l$. Then, we move it $z-\dist(l,v)$ steps 
down towards $q$ from $l$. In this case, the target node is an ancestor of $q$ and the distance from the root is $d=dist(1,l)+z-\dist(l,v)$. So, we can find it using the result of  $\textsc{LA}(v,dist(1,l)+z-\dist(l,v))$ procedure. The algorithm is presented in Algorithm \ref{alg:move}.

  \begin{algorithm}[ht]
    \caption{$\textsc{Move}(v,z)$. Moves of a server from $v$ to distance $z$ on a path from $v$ to $q$.} \label{alg:move}
    \begin{algorithmic}
        \State $l=\textsc{LCA}(v,q)$
        \If{ $z\leq \dist(l,v)$}
        \State $Result \gets \textsc{LA}(v,dist(1,v)-z)$
        \EndIf
       \If{ $z>\dist(l,v)$}
        \State $Result \gets \textsc{LA}(q,dist(1,l)+z-\dist(l,v))$
        \EndIf
        \State \Return{$Result$}
    \end{algorithmic}
\end{algorithm}

\begin{lemma}\label{lm:move}
    The time complexity of \textsc{Move} is $O(1)$.
\end{lemma}
\begin{proof}
The time complexity of the $\textsc{LA}$ procedure is $O(1)$  due to Lemma \ref{lm:la}. The time complexity of the $\textsc{LCA}$ procedure is $O(1)$  due to Lemma \ref{lm:lca}. Furthermore, we can compute the distance between any two nodes in $O(1)$ thanks to the preprocessing. Therefore, the total complexity is $O(1)$.
\end{proof}

\begin{theorem}\label{th:query-proc}
The time complexity of the query processing phase is $O\left( k\log k\right)$.
\end{theorem}
\begin{proof}
The complexity of constructing a virtual tree is $O(k \log k)$. The complexity of the two depth-first search algorithms is $O(k\log k)$. The complexity of moving all servers is $O(k)$. So, the total complexity is $O(k \log k)$.
\end{proof}
\vspace{-0.6cm}
\section{Conclusion}
\vspace{-0.3cm}
In the current work, we present a new algorithm with $O(n)$ time complexity for preprocessing and $O(k\log k)$ for processing a query. Notably, other existing algorithms have worse complexities. Additionally, we proposed that the query processing complexity does not depend on $n$. The complexity of preprocessing seems the best possible. 
At the same time, the existence of an algorithm with $O(k)$ query processing time complexity is still an open question.

%
%
%
\bibliographystyle{splncs04}
\bibliography{biblio}
\newpage
\appendix
\section{Implementation of \textsc{ComputeDistance}() procedure}\label{apx:dist}
The procedure is based on the DFS algorithm. The base recursive algorithm is presented as a procedure $\textsc{ComputeDistanceBase}()$ in Algorithm \ref{alg:distBase}, and the whole algorithm is presented in Algorithm \ref{alg:dist}. 

\begin{algorithm}[ht]
    \caption{$\textsc{ComputeDistanceBase}(v)$. The base recursive procedure for processing a node $v$. We assume that we have a global $timer$ integer variable.} \label{alg:distBase}
    \begin{algorithmic}
   
        \State $t_{in}(v)\gets timer$
        \State $timer\gets timer+1$
        \For{$u\in \textsc{Children}(v)$}
        \State $dist(1,u)\gets dist(1,v)+1$
        \State $\textsc{ComputeDistanceBase}(u)$
        \EndFor
            \State $t_{out}(v)\gets timer$
        \State $timer\gets timer+1$
    \end{algorithmic}
\end{algorithm}

\begin{algorithm}[ht]
    \caption{$\textsc{ComputeDistance}(v)$. The distance computing algorithm.} \label{alg:dist}
    \begin{algorithmic}
   
        \State $timer\gets 0$
        \State $dist(1,1)\gets 0$
        \State $\textsc{ComputeDistanceBase}(1)$
    \end{algorithmic}
\end{algorithm}
\end{document}